\documentclass[runningheads]{llncs}
\usepackage{xcolor,latexsym,amssymb,stmaryrd,graphicx,hyperref}
\usepackage{amsmath,amssymb}

\newcommand{\bA}{\ensuremath{\mathfrak{A}}}
\newcommand{\Fresse}{Fra\"{i}ss\'e}
\newcommand{\bB}{\ensuremath{\mathfrak{B}}}
\newcommand{\bC}{\ensuremath{\mathfrak{C}}}
\newcommand{\bD}{\ensuremath{\mathfrak{D}}}

\newcommand{\red}[1]{#1}
\DeclareMathOperator{\Age}{Age}
\DeclareMathOperator{\Aut}{Aut}
\DeclareMathOperator{\ar}{ar}
\DeclareMathOperator{\Forbe}{Forb}
\DeclareMathOperator{\Betw}{Betw}

\DeclareMathOperator{\Csp}{CSP}
\DeclareMathOperator{\Forb}{Forb}

\title{ASNP: 
a tame fragment of existential second-order logic}
\author{Manuel Bodirsky, Simon Kn\"auer, Florian Starke}
\institute{Institute of Algebra, TU Dresden} 

\begin{document}
\maketitle

\begin{abstract}
\emph{Amalgamation SNP (ASNP)} is a fragment of existential second-order logic that strictly contains binary connected MMSNP of Feder and Vardi and binary guarded monotone SNP of Bienvenu, ten Cate, Lutz, and Wolter; it is a promising candidate for an expressive subclass of NP that exhibits a complexity dichotomy. We show that 
ASNP has 
 a complexity dichotomy if and only 
if the infinite-domain dichotomy conjecture holds 
for constraint satisfaction problems 
for \red{first-order} reducts of binary finitely bounded homogeneous structures. For such CSPs, powerful universal-algebraic hardness conditions are known that are conjectured to describe the border between NP-hard and polynomial-time tractable CSPs. 
The connection to CSPs also implies that
every ASNP sentence can be
evaluated in polynomial time on classes of finite structures of bounded treewidth. 
We show that the syntax of 
ASNP is decidable. 
The proof relies on the fact that for classes of finite binary structures given by finitely many forbidden substructures, the amalgamation property is decidable.  
\end{abstract}
  
  
 
 
 
 \section{Introduction}
 Feder and Vardi in their groundbreaking work~\cite{FederVardi} formulated the 
 famous \emph{dichotomy conjecture} for finite-domain constraint satisfaction problems, which has recently been resolved~\cite{BulatovFVConjecture,ZhukFVConjecture}. Their motivation to study finite-domain CSPs was the question which
 fragments of existential second-order logic might exhibit a complexity dichotomy in the sense
 that every problem that can be expressed in the fragment is either in P or NP-complete. 
 Existential second-order logic without any restriction is known to capture NP~\cite{Fagin} 
 and hence does not have a complexity dichotomy by an old result of Ladner~\cite{Ladner}. Feder and Vardi proved
 that even the fragments of \emph{monadic SNP} and \emph{monotone SNP}
 do not have a complexity dichotomy since
 every problem in NP is polynomial-time equivalent to a problem that can be expressed in these fragments. However, the dichotomy for finite-domain CSPs implies that \emph{monotone monadic SNP (MMSNP)} has a dichotomy, too~\cite{FederVardi,Kun}.
  
 MMSNP is also known to have a 
 tight connection to a certain class of infinite-domain CSPs~\cite{MMSNP}: an MMSNP sentence is equivalent to a \emph{connected} MMSNP sentence if and only if it describes an infinite-domain CSP. Moreover, every 
 problem in MMSNP is equivalent to a finite disjunction of connected MMSNP sentences. 
   The infinite structures that appear in this connection are tame from a model-theoretic perspective: they are reducts of finitely bounded homogeneous structures (see Section~\ref{sect:amalgamation}). CSPs for
 such structures are believed to have a complexity dichotomy, too; there is even a 
 known hardness condition
 such that all other CSPs in the class are conjectured to be in P~\cite{BPP-projective-homomorphisms}. The hardness condition can be expressed in several equivalent forms~\cite{BartoPinskerDichotomy,BKOPP-equations}. 
    
In this paper we investigate another candidate for 
an expressive logic that has a complexity dichotomy. Our minimum requirement for what constitutes a \emph{logic} is relatively liberal:
 we require that the syntax of the logic should
 be decidable. The same requirement has been
 made for the question whether there exists a logic that captures the class of \red{polynomial-time solvable decision problems} (see, e.g.,~\cite{Gurevich,GroheLogicforP}). 
 The idea of our logic is to modify monotone SNP
so that only CSPs for model-theoretically tame
structures can be expressed in the logic;
the challenge is to come up with a definition
of such a logic which has a decidable syntax. 
We would like to require that the 
(universal) first-order part of a monotone SNP sentence describes an \emph{amalgamation class}. We mention that the \emph{Joint Embedding Property (JEP)}, which follows from the \emph{Amalgamation Property (AP)}, 
has recently been shown to be undecidable~\cite{BraunfeldUndec}. 
In contrast, we use the fact that the AP for binary signatures is decidable 
(Section~\ref{sect:amalgamation-decidable}). 
We call our new logic \red{\emph{Amalgamation SNP (ASNP)}}.
 This logic contains binary connected MMSNP;
 it also contains the more expressive logic of \emph{binary connected guarded monotone SNP}.
 Guarded monotone SNP (GMSNP) has been introduced in the context
 of knowledge representation~\cite{OBDA} (see Section~\ref{sect:gmsnp}). Every problem
 that can be expressed in ASNP or in connected GMSNP is a CSP for some countably infinite $\omega$-categorical structure $\bB$\@. 
 In Section~\ref{sect:btw} we present an example application  of this fact: every problem that can be expressed in one of these logics can be solved in polynomial time on instances of bounded treewidth.

 \section{Constraint Satisfaction Problems}
Let $\bA,\bB$ be structures with a finite relational signature $\tau$; 
each symbol $R \in \tau$ is equipped with an \emph{arity} $\ar(R) \in {\mathbb N}$. 
 A function $h \colon A \to B$ is called a \emph{homomorphism from $\bA$ to $\bB$} if 
 for every $R \in \tau$ and $(a_1,\dots,a_{\ar(R)}) \in R^{\bA}$ we have $(h(a_1),\dots,h(a_{\ar(R)})) \in R^{\bB}$; in this case we write $\bA \to \bB$.
 We write $\Csp(\bB)$ for the class of all finite
 $\tau$-structures $\bA$ such that $\bA \to \bB$.
 
 \begin{example}
 If $\bB = K_3$ is the 3-clique, i.e., the complete undirected graph with three vertices, then $\Csp(\bB)$ is the graph 3-colouring problem,
 which is NP-complete~\cite{GareyJohnson}. 
 \end{example}
 
 \begin{example}\label{expl:q1}
 If $\bB = ({\mathbb Q};<)$ 
 then $\Csp(\bB)$ is the digraph acyclicity problem, which is in P. 
 \end{example}
 
 \begin{example}\label{expl:betw-1}
 If $\bB = ({\mathbb Q};\Betw)$ for $\Betw := \{(x,y,z) \mid x < y < z \vee z < y < x \}$  
 then $\Csp(\bB)$ is the Betweenness problem, which is NP-complete~\cite{GareyJohnson}. 
\end{example}

 The
\emph{union} of two $\tau$-structures $\bA,\bB$ is the $\tau$-structure
$\bA\cup\bB$ with domain $A \cup B$ and the relation
$R^{\bA\cup\bB} := R^{\bA}\cup R^{\bB}$ for every  $R\in\tau$. 
The \emph{intersection} $\bA \cap \bB$ is defined analogously. 
A \emph{disjoint
  union} of $\bA$ and $\bB$ is the union of isomorphic copies of
$\bA$ and $\bB$ with disjoint domains. 
As disjoint unions are unique up
to isomorphism, we usually speak of \emph{the} disjoint union of $\bA$ and
$\bB$, and denote it by $\bA \uplus \bB$. 
 A structure is \emph{connected} if it cannot be written as a disjoint union of at least two structures 
 with non-empty domain. 
 A class of structures $\mathcal C$ is \emph{closed under inverse homomorphisms} if whenever $\bB \in {\mathcal C}$ and $\bA$ homomorphically maps to $\bB$ we have $\bA \in {\mathcal C}$. 
 If $\tau$ is a finite relational signature, then 
it is well-known and easy to see~\cite{Bodirsky-HDR} that ${\mathcal C} = \Csp(\bB)$ for a countably infinite $\tau$-structure $\bB$ if and only if $\mathcal C$ is closed under inverse homomorphisms and disjoint unions. 
  
\section{Monotone SNP} 
Let $\tau$ be a finite relational signature,
i.e., $\tau$ is a set of relation symbols $R$,
each equipped with an \emph{arity} $\ar(R) \in {\mathbb N}$. 
An \emph{SNP ($\tau$-) sentence} is an existential second-order ($\tau$-) sentence with a universal 
first-order part, i.e., 
a sentence of the form
$$ \exists R_1,\dots,R_k \, \forall x_1,\dots,x_n \colon \phi$$
where $\phi$ is a quantifier-free formula over the signature $\tau \cup \{R_1,\dots,R_k\}$. 
We make the additional convention that the equality symbol, which is usually allowed in first-order logic, is not allowed in $\phi$
(see~\cite{FederVardi}). 
We write $\llbracket \Phi \rrbracket$ for the class of all finite models of $\Phi$. 

\begin{example}\label{expl:q2}
$\Csp({\mathbb Q};<) =  \llbracket \Phi \rrbracket$ for the
SNP $\{<\}$-sentence $\Phi$ given below. 
\begin{align*}
\exists T \, \forall x,y,z \big ( & (\neg (x<y) \vee  T(x,y)) \\
\wedge & \big ( \neg T(x,y) \vee \neg T(y,z)  \vee T(x,z) \big) \wedge \neg T(x,x) \big )
\end{align*}
\end{example}

A class $\mathcal C$ of finite $\tau$-structures is said to be \emph{in SNP} if there exists an SNP $\tau$-sentence
$\Phi$ such that $\llbracket \Phi \rrbracket = {\mathcal C}$; we use analogous definitions for
all logics considered in this paper. 
We may assume that the quantifier-free part 
of SNP sentences is written in conjunctive normal form, and then use the usual terminology (\emph{clauses}, \emph{literals}, etc). 

\begin{definition}\label{def:mmc}
An SNP $\tau$-sentence $\Phi$ with quantifier-free part $\phi$ and existentially quantified relation symbols $\sigma$ 
 is called 
\begin{itemize}
\item \emph{monotone} 
if each literal of $\phi$ with a symbol from $\tau$ is \emph{negative}, i.e., of the form
$\neg R(\bar x)$ for $R \in \tau$. 
\item \emph{monadic} if all the existentially
quantified relations are unary. 
\item \emph{connected} if each clause of $\phi$
is connected, i.e., the following $\tau \cup \sigma$-structure $\bC$ is connected: 
the domain of $\bC$ is the set of variables of the clause, and $t \in R^{\bC}$ if and only if $\neg R(t)$ is a disjunct of the clause. 
\end{itemize}
\end{definition}
The SNP sentence from Example~\ref{expl:q2} 
is monotone, but not monadic, and it can be shown that there does not exist an equivalent MMSNP sentence~\cite{Bodirsky}. 
The following
is taken from \cite{Bodirsky-HDR} and a proof can be found in Appendix~\ref{sect:a1} for the convenience of the reader. 

\begin{theorem}
\label{thm:connected}
Every sentence in connected monotone SNP
describes a problem of the form
$\Csp(\bB)$ for some relational structure $\bB$. 
Conversely, for every structure $\bB$, 
if $\Csp(\bB)$ is in SNP then it is also
in connected monotone SNP. 
\end{theorem}


\section{Amalgamation SNP}
\red{In this section we define the new logic \emph{Amalgamation SNP (ASNP)}.}
We first revisit some basic concepts from model theory. 

\subsection{The Amalgamation Property}
\label{sect:amalgamation}
Let $\tau$ be a finite relational signature
and let $\mathcal C$ be a class of $\tau$-structures. 
We say that $\mathcal C$ is \emph{finitely bounded} if there exists a finite set of finite $\tau$-structures ${\mathcal F}$
such that $\bA \in {\mathcal C}$ if and only if no structure in $\mathcal F$ embeds into $\bA$;
in this case we also write $\mathcal C = \Forb(\bA)$. 
Note that $\mathcal C$ is finitely bounded if and only if there exists a universal $\tau$-sentence
$\phi$ (which might involve the equality symbol) such that for every finite $\tau$-structure $\bA$ we have $\bA \models \phi$ if and only if $\bA \in {\mathcal C}$. We say that $\mathcal C$ has 
\begin{itemize}
\item the \emph{Joint Embedding Propety (JEP)}
if for all structures $\bB_1,\bB_2 \in {\mathcal C}$
there exists a structure $\bC \in {\mathcal C}$ 
that embeds both $\bB_1$ and $\bB_2$. 
\item the \emph{Amalgamation Property (AP)}
if for any two structures $\bB_1,\bB_2 \in {\mathcal C}$ such that $B_1 \cap B_2$ induce the same substructure in $\bB_1$ and in $\bB_2$ (a so-called \emph{amalgamation diagram}) 
 there exists a structure $\bC \in {\mathcal C}$
and embeddings $e_1 \colon \bB_1 \hookrightarrow \bC$ and $e_2 \colon \bB_2 \hookrightarrow \bC$ such that $e_1(a) = e_2(a)$ for all $a \in B_1 \cap B_2$. 
\end{itemize}
Note that since $\tau$ is relational, 
the AP implies the JEP. 
A class of finite $\tau$-structures which has the AP and is closed under induced substructures and isomorphisms is called an \emph{amalgamation class}. 

The \emph{age} of $\bB$ is the class of all finite $\tau$-structures that embed into $\bB$. 
We say that $\bB$ is \emph{finitely bounded} if 
$\Age(\bB)$ is finitely bounded. 
A relational $\tau$-structure $\bB$ is called \emph{homogeneous} if every isomorphism
between finite substructures of $\bB$ can
be extended to an automorphism of $\bB$. 
\Fresse's theorem implies that for every amalgamation class $\mathcal C$ there exists a countable homogeneous $\tau$-structure $\bB$ 
with $\Age(\bB) = {\mathcal C}$; the structure
$\bB$ is unique up to isomorphism, also called the \emph{\Fresse-limit} of $\mathcal C$. 
Conversely, it is easy to see that the age of a homogeneous $\tau$-structure is an amalgamation class. A structure $\bA$ is called 
 a \emph{reduct} of a structure $\bB$
if $\bA$ is obtained from $\bB$ by restricting the signature. It is called \red{a \emph{first-order reduct} of $\bB$ if $\bA$ is obtained from $\bB$ by
first expanding by all first-order definable relations, and then restricting the signature.}
An example of a first-order reduct of 
$({\mathbb Q};<)$ 
is the structure $({\mathbb Q};\Betw)$ 
from Example~\ref{expl:betw-1}. 

\subsection{Defining Amalgamation SNP} 
As we have mentioned in the introduction, 
the idea of our logic is to require that \red{a certain class
of finite structures associated to the 
first-order part of an SNP sentence} is an amalgamation class. 
We then use the fact that for binary signatures, the amalgamation property is decidable (Section~\ref{sect:amalgamation-decidable}).  

\begin{definition}\label{def:asnp}
Let $\tau$ be a finite relational signature. 
An \emph{Amalgamation SNP $\tau$-sentence} is an SNP sentence $\Phi$ \red{of the form $\exists R_1,\dots,R_k \, \forall x_1,\dots,x_n \colon \phi$ where}
\begin{itemize}
\item $R_1,\dots,R_k$ are binary; 
\item $\phi$ is a conjunction of $\{R_1,\dots,R_k\}$-formulas and of conjuncts of the form 
$S(x_1,\dots,x_k) \Rightarrow \psi(x_1,\dots,x_k)$
where $S \in \tau$ and $\psi$ is a $\{R_1,\dots,R_k\}$-formula; 
\item the class of $\{R_1,\dots,R_k\}$-reducts 
of the finite models of $\phi$ is an amalgamation class.
\end{itemize}
\end{definition}
\red{Note that ASNP
inherits from SNP the restriction that equality symbols are not allowed.
Also note that Amalgamation SNP sentences
are necessarily monotone. This implies in particular that the class of $\{R_1,\dots,R_k\}$-reducts of the finite models of $\phi$ is precisely the class of finite $\{R_1,\dots,R_k\}$-structures that satisfy the conjuncts of $\phi$ that are $\{R_1,\dots,R_k\}$-formulas (i.e., that do not contain any symbol from $\tau$). }

\begin{example}\label{expl:q3}
The monotone SNP sentence from Example~\ref{expl:q2} describing  
$\Csp({\mathbb Q};<)$ is in ASNP.
The problem $\Csp({\mathbb Q};\Betw)$
from Example~\ref{expl:betw-1}
can be expressed by the ASNP\@ sentence
\begin{align*}
\exists T \, \forall x,y,z \big ( & (\Betw(x,y,z) \Rightarrow ((T(x,y) \wedge T(y,z)) \vee (T(z,y) \wedge T(y,x)) ) \\
\wedge & \big ( (T(x,y) \wedge T(y,z)) \Rightarrow T(x,z) \big) \wedge \neg T(x,x) \big ) \, .
\end{align*}
Also note that every finite-domain CSP can be expressed in ASNP.
\end{example}

\red{If $\rho$ are the existentially quantified binary
relations of an ASNP sentence $\Phi$, then 
the class of finite models of 
the first-order part of $\Phi$ has the JEP, 
and since equality
is not allowed in SNP the class is even closed under disjoint unions;
it follows that also $\Phi$ is closed under disjoint unions. It can be shown 
as in the proof of Theorem~\ref{thm:connected}
that every Amalgamation SNP sentence 
can be rewritten into an equivalent connected Amalgamation SNP sentence.}

\subsection{ASNP and CSPs}
We present the link between ASNP and infinite-domain CSPs. 

\begin{theorem}\label{thm:first}
For every ASNP $\tau$-sentence $\Phi$ there
exists a \red{first-order} reduct $\bC$ of a binary finitely bounded homogeneous structure such that
$\Csp(\bC) = \llbracket \Phi \rrbracket$. 
\end{theorem}
\begin{proof}
Let $\rho$ be the set of existentially quantified relation symbols of $\Phi$. 
Let $\phi = \forall x_1,\dots,x_n \colon \psi$, for a quantifier-free formula $\psi$ in conjunctive normal form, be the first-order part of $\Phi$.
Let $\mathcal C$ be the class of $\rho$-reducts 
of the finite models of $\phi$; by assumption, 
$\mathcal C$ is an amalgamation class. 
Moreover, $\mathcal C$ is finitely bounded because it is the class of models of a universal $\rho$-sentence. 
Let $\bB$ be the \Fresse-limit of $\mathcal C$; then $\bB$ is a finitely bounded homogeneous structure. 
\red{Let $\bC$ be the $\tau$-structure which is the
first-order reduct of the structure $\bB$ where
the relation $S^{\bC}$ for $S \in \tau$ is defined as follows:
if $\phi_1,\dots,\phi_s$ are all the $\rho$-formulas
such that $\psi$ contains the conjunct
$S(x_1,\dots,x_k) \Rightarrow \phi_i(x_1,\dots,x_k)$ for all $i \in \{1,\dots,s\}$, then the first-order definition of $S$ is 
given by $S(x_1,\dots,x_k) \Leftrightarrow (\phi_1 \wedge \dots \wedge \phi_s)$.}


\medskip
{\bf Claim 1.} If $\bA$ is a finite $\tau$-structure  such that $\bA \to \bC$, then $\bA \models \Phi$.

Let $h \colon \bA \to \bC$ be a homomorphism. 
Let $\bA'$ be the $(\tau \cup \rho)$-expansion of $\bA$ where $R \in \rho$ of arity $l$ denotes $\{(a_1,\dots,a_l) \mid (h(a_1),\dots,h(a_l)) \in R^{\bB} \}$. 
Then 
$\bA'$ satisfies $\phi$:
to see this, let $a_1,\dots,a_n \in A$ and
let $\psi'$ be a conjunct of $\psi$. 
Since $\bC \models \forall x_1,\dots,x_n \colon \psi$ we have in particular that $\bC \models \psi'(h(a_1),\dots,h(a_n))$ and so there must be a disjunct $\psi''$ of $\psi'$ such that 
$\bC \models \psi''(h(a_1),\dots,h(a_n))$. 
Then one of the following cases applies. 
\begin{itemize}
\item $\psi''$ is a $\tau$-literal and hence must be negative since $\Phi$ is a monotone SNP sentence. In this case $\bC \models \psi''(h(a_1),\dots,h(a_n))$
implies $\bA' \models \psi''(a_1,\dots,a_n)$ since $h$ is a homomorphism.  
\item $\psi''$ is a $\rho$-literal. Then by the definition of $\bA'$ we have that
$\bA' \models \psi''(a_1,\dots,a_n)$ 
if and only if $\bC \models \psi''(h(a_1),\dots,h(a_n))$. 
\end{itemize}
Hence, $\bA' \models \psi'(a_1,\dots,a_n)$. 
Since the conjunct $\psi'$ of $\psi$
and $a_1,\dots,a_n \in A$ were arbitrarily chosen, we have that $\bA' \models \forall x_1,\dots,x_n \colon \psi$. Hence, $\bA$ satisfies $\Phi$. 

\medskip
{\bf Claim 2.} If $\bA$ is a finite $\tau$-structure such that $\bA \models \Phi$, then $\bA \to \bC$. 

If $\bA$ has an expansion $\bA'$ that satisfies $\phi$, then there exists an embedding from $\bA'$ into $\bB$ by the definition of $\bB$. This embedding is in particular a homomorphism from $\bA$ to $\bC$. \qed
\end{proof} 


The proof of the following theorem can be found in Appendix~\ref{sect:csp-ASNP}. 

\begin{theorem}\label{thm:find-asnp}
Let $\bC$ be a \red{first-order} reduct of a binary finitely bounded homogeneous structure $\bB$.
Then $\Csp(\bC)$ can be expressed in ASNP. 
\end{theorem}

\begin{corollary}
ASNP has a complexity dichotomy
if and only if the infinite-domain dichotomy conjecture is true for \red{first-order} reducts of binary finitely bounded homogeneous structures. 
\end{corollary}

 
 \section{Deciding Amalgamation}
\label{sect:amalgamation-decidable}
In this section we show how to algorithmically decide whether
a given existential second-order sentence is
in ASNP. 
The following is a known fact in the model theory of homogeneous structures (the first author has learned the fact from Gregory Cherlin), but we are not aware of any published proof in the literature. 

\begin{theorem}\label{thm:decide-amalg}
Let ${\mathcal F}$ be a finite set of finite binary relational $\tau$-structures. 
There is an algorithm that decides whether
$\Forbe({\mathcal F})$ has the amalgamation property. 
\end{theorem}

\begin{proof}
Let $m$ be the maximal size of a structure in $\mathcal F$, and let $\ell$ be the number of isomorphism types of two-element structures in $\Age(\bB)$. It is well-known and easy to prove
that ${\mathcal C} := \Forbe({\mathcal F})$ has the amalgamation property
if and only if it has the so-called \emph{1-point amalgamation property}, i.e., 
the amalgamation property restricted to diagrams
$(\bB_1,\bB_2)$ where $|B_1| = |B_2| = |B_1 \cap B_2|$. 
Suppose that $(\bB_1,\bB_2)$
is such an amalgamation diagram without amalgam. 
Let $B_0 := B_1 \cap B_2$.
Let $B_1 \setminus B_0 = \{p\}$ and $B_2 \setminus B_0 = \{q\}$. 
Let $\bD$ be a $\tau$-structure $\bD$ with domain
$B_1 \cup B_2$ such that $\bB_1$ and $\bB_2$ are substructures of $\bD$.
Since $\bD$ by assumption is not 
an amalgam for $(\bB_1,\bB_2)$,
there must exist 
$A = \{a_1,\dots,a_{m-2}\} \in B_0$ 
such that the substructure of $\bD$ 
induced by $\{a_1,\dots,a_{m-2},p,q\}$
embeds a structure from $\mathcal F$. 

Note that the number of such $\tau$-structures $\bD$ is bounded by $\ell$ since they
only differ by the substructure induced by $p$ and $q$. So let $A_1,\dots,A_{\ell} \subseteq B_0$ be a list of sets witnessing that all of these structures $\bD$ embed a structure from $\mathcal F$. 
Let $\bC_1$ be the substructure of $\bB_1$ induced by $\{p\} \cup A_1 \cup \cdots \cup A_{\ell}$ and $\bC_2$ be the substructure of $\bB_2$ induced by $\{q\} \cup A_1 \cup \dots \cup A_{\ell}$. Suppose for contradiction that $(\bC_1,\bC_2)$ has an amalgam $\bC$; we may assume that 
this amalgam is of size at most $(m-2) \cdot \ell$. 
 Depending on the two-element 
structure induced by $\{p,q\}$ in $\bC$, there 
exists an $i \leq \ell$ such that the structure induced by $\{p,q\} \cup A_i$ in $\bC$ embeds a structure from $\mathcal F$, a contradiction. 
\qed 
\end{proof}

\begin{corollary}
There is an algorithm that decides for a given existential second-order sentence $\Phi$ whether it is in ASNP. 
\end{corollary}
\begin{proof}
Let $k$ be the maximal number of variables per clause in the first-order part $\phi$ of $\Phi$, and
let $\mathcal F$ be the set of all structures at most  
the elements $\{1,\dots,k\}$ that do not satisfy $\phi$. Then $\Forb({\mathcal F}) = \llbracket \phi \rrbracket$ and the result follows from Theorem~\ref{thm:first}. 
\end{proof}

 \section{Guarded Monotone SNP}
 \label{sect:gmsnp}
 In this section we revisit an expressive 
 generalisation of MMSNP introduced by Bienvenu, ten Cate, Lutz, and Wolter~\cite{OBDA} in the context of ontology-based data access, called
 \emph{guarded monotone SNP (GMSNP)}. 
 It is equally expressive as the logic MMSNP$_2$ introduced by Madelaine~\cite{MadelaineUniversal}\footnote{MMSNP$_2$ relates to MMSNP as Courcelle's MSO$_2$ relates to MSO~\cite{EngelfrietCourcelle}.}. 
 We will see that every GMSNP sentence is equivalent to a finite disjunction of
 \emph{connected} GMSNP sentences (Proposition~\ref{prop:disconnected}),
 each of which lies in ASNP  if the signature is binary (Theorem~\ref{thm:gmsnp}). 
 
 \begin{definition}
A monotone SNP $\tau$-sentence $\Phi$ with existentially quantified relations $\rho$ is called 
\emph{guarded} if each conjunct of $\Phi$ 
can be written in the form
\begin{align*}
\alpha_1 \wedge \cdots \wedge \alpha_n \Rightarrow \beta_1 \vee \cdots \vee \beta_m,
\quad \text{where}
\end{align*}
\begin{itemize}
\item $\alpha_1,\dots,\alpha_n$ are atomic $(\tau \cup \rho)$-formulas, called \emph{body atoms},
\item $\beta_1,\dots,\beta_m$ are 
atomic $\rho$-formulas, called \emph{head atoms},
\item for every head atom $\beta_i$ there is a body atom $\alpha_j$ such that $\alpha_j$ contains all variables from $\beta_i$
(such clauses are called \emph{guarded}).  
\end{itemize}
We do allow the case that $m=0$, i.e., the case where the head consists of the empty disjunction,  which is equivalent to $\perp$ (false). 
\end{definition}


Our next proposition is well-known for MMSNP
and can be extended to guarded SNP, too.
See Appendix~\ref{sect:a4} for the proof. 

\begin{proposition}\label{prop:disconnected}
Every GMSNP sentence $\Phi$ is equivalent to a finite disjunction $\Phi_1 \vee \cdots \vee \Phi_k$ of connected GMSNP sentences. 
\end{proposition}

It is well-known and easy to see~\cite{FoniokNesetril} that 
each of $\Phi_1,\dots,\Phi_k$ can be reduced to $\Phi$ in polynomial time. Conversely, if each
of $\Phi_1,\dots,\Phi_k$ is in P, then
$\Phi$ is in P, too. 
It follows in particular that if connected GMSNP has a complexity dichotomy into P and NP-complete, then so has GMSNP. 
 
 \begin{theorem}\label{thm:gmsnp}
For every sentence $\Phi$ in connected GMSNP 
there exists a reduct $\bC$ of a finitely bounded
homogeneous structures such that $\llbracket \Phi \rrbracket = \Csp(\bC)$. If all \red{existentially quantified} relation symbols in $\Phi$ are binary then it is equivalent to an ASNP sentence.  
\end{theorem}

In the proof of Theorem~\ref{thm:gmsnp} we use
a result of Cherlin, Shelah, and Shi~\cite{CherlinShelahShi} in a strengthened form due
to Hubi\v{c}ka and Ne\v{s}et\v{r}il~\cite{Hubicka-Nesetril}, namely
that for every finite set $\mathcal F$ of finite $\sigma$-structures, for some finite relational signature $\sigma$, 
there exists a finitely bounded homogeneous $(\sigma \cup \rho)$-structure $\bB$ such that 
a finite $\sigma$-structure $\bA$ homomorphically maps to $\bB$ if none of the structures in $\mathcal F$ 
homomorphically maps to $\bB$.
We now prove Theorem~\ref{thm:gmsnp}. 

\begin{proof}
Let $\Phi$ be a $\tau$-sentence in connected guarded monotone SNP with existentially quantified relation symbols $\{E_1,\dots,E_k\}$. 
Let $\sigma$ be the signature which 
contains for every relation symbol $R \in \{E_1,\dots,E_k\}$ two new relation symbols $R^+$ and $R^-$ of the same arity and for
every relation symbol $R \in \tau$ a new relation
symbol $R'$. 
Let $\phi$ be the first-order part of $\Phi$, written in conjunctive normal form, and let $n$ be the number of variables in the largest clause of $\phi$. 
Let $\phi'$ be the sentence obtained from $\phi$ by replacing each occurrence of $R \in \{E_1,\dots,E_k\}$ by $R^+$ and 
each occurrence of $\neg R$ by $R^-$, and 
finally each occurrence of $R \in \tau$ by
$R'$. 
Let $\mathcal F$ be the (finite) class of all finite 
$\sigma$-structures with at most $n$ elements that do not satisfy $\phi'$. 
We apply the mentioned theorem of Hubi\v{c}ka and Ne\v{s}et\v{r}il to $\mathcal F$, and obtain a finitely bounded homogeneous $\sigma \cup \rho$-structure $\bB$ such that the age of the
  $\sigma$-reduct $\bC$ of $\bB$ equals $\Forbe({\mathcal N})$. 
  	We say that $S \subseteq B$ is
	\emph{correctly labelled} if
	for every $R \in \{E_1,\dots,E_k\}$ of arity $m$ and
	$s_1,\dots,s_m \in S$ we have
	$R^-(s_1,\dots,s_m)$ if and only if 
	$\neg R(s_1,\dots,s_m)$. 	
	Let $\bB'$ the $\tau \cup \sigma \cup \rho$-expansion of $\bB$
	where $R \in \tau$ of arity $m$ denotes 
	$$\{(t_1,\dots,t_m) \in (R')^{\bB} \mid 
	\{t_1,\dots,t_m\} \text{ is correctly labelled}\}.$$
		Since $\bB$ is finitely bounded homogeneous, $\bB'$ is finitely bounded homogeneous, too. 
Let $\bC$ be the $\tau$-reduct of $\bB'$. 
  We claim that $\llbracket \Phi \rrbracket = \Csp(\bC)$. First suppose that $\bA$ is a finite $\tau$-structure that satisfies $\Phi$. Then it
  has an $\{E_1,\dots,E_k\}$-expansion $\bA'$ that satisfies $\phi$. Let $\bA''$ be the $\sigma$-structure with the same domain as $\bA'$
  where  
  \begin{itemize}
  \item $R'$ denotes $R^{\bA'}$ for each $R \in \tau$;
  \item $R^+$ denotes $R^{\bA'}$ for each $R \in \{E_1,\dots,E_k\}$;
  \item $R^{-}$ denotes $\neg R^{\bA'}$ for each $R \in \{E_1,\dots,E_k\}$. 
  \end{itemize}
  Then $\bA''$ satisfies $\phi'$, and hence embeds into $\bB$. This embedding
  is a homomorphism from $\bA$ to $\bC$
  since the image of the embedding is correctly labelled by the construction of $\bA''$. 
  
  Conversely, suppose that $\bA$ has a homomorphism $h$ to $\bC$. 
  Let $\bA'$ be the $\tau \cup \{E_1,\dots,E_k\}$-expansion of $\bA$ by defining
  $(a_1,\dots,a_n) \in R^\bA$ if and only if
  $(h(a_1),\dots,h(a_n)) \in R^{\bB'}$, for every $n$-ary $R \in \{E_1,\dots,E_k\}$.
  Then each clause
  of $\phi$ is satisfied,
  because each clause of $\phi$ is guarded: 
  let $x_1,\dots,x_m$ be the variables of some
  clause of $\phi$. 
  If 
  $a_1,\dots,a_m \in A$ 
   satisfy the body of this clause,
  and $\psi(a_{i_1},\dots,a_{i_l})$ is a head atom
  of such a clause, then the set $\{h(a_{i_1}),\dots, h(a_{i_l})\}$ is correctly labelled.
   This implies
   that some of the head atoms of the clause 
   must be true in $\bA'$ 
   because $\bB'$ satisfies $\phi'$. 
   The second statement follows from Theorem~\ref{thm:find-asnp}. 
\qed 
\end{proof}

The following example shows that GMSNP does
not contain ASNP. 
 
 \begin{example}
  $\Csp({\mathbb Q};<)$ is in ASNP (see Example~\ref{expl:q3}) but not
  in GMSNP.  
  Indeed, suppose that $\Phi$ is a GMSNP sentence which is true on all finite directed paths. 
We assume that the quantifier-free part $\phi$ of 
$\Phi$ is in conjunctive normal form.
Let $\rho$ be the existentially quantified relation symbols of $\Phi$, let $k := |\rho|$, 
and let $l$ be the number of variables in $\Phi$. 
A directed path of length $(2^{2kl}+1) l$, viewed as a $\{<\}$-structure, 
 satisfies $\Phi$, and therefore it has an $\{<\} \cup \rho$-expansion $\bA$ that satisfies $\phi$. Note that there are $L := 2^{2kl}$ different $\{<\} \cup \rho$-expansions of a path of length $l$ (for each vertex and each edge of the path we have to decide which of the $k$ predicates holds), 
 and hence there must be $i,j \in \{0,\dots, L\}$ with $i<j$ such that the substructures of $\bA$
 induced by $i l +1,i l+2,\dots,i l +l$
 and by $j l + 1,j l+2,\dots,j l+l$ are isomorphic. 
We then claim that the directed cycle
$(i+1) l + 1,(i+1) l + 2,\dots, j l + 1,\dots, jl + l, (i+1) l + 1$ satisfies $\Phi$: this is witnessed by 
the $\{<\} \cup \rho$-expansion inherited from $\bA$ which satisfies $\phi$. Hence, $\Phi$ does not express 
digraph acyclicity. \qed
 \end{example}

\section{Application: Instances of Bounded Treewidth}
\label{sect:btw}
If a computational problem can be formulated 
in ASNP or in GMSNP, then this has remarkable consequences besides a potential complexity dichotomy. In this section we show that 
every problem that can be formulated in ASNP
or in GMSNP is in P when restricted to 
instances of bounded treewidth. 
The corresponding result for Monadic Second-Order Logic (MSO) instead of ASNP is a famous theorem of Courcelle~\cite{EngelfrietCourcelle}. 
We strongly believe that ASNP is not contained in MSO (consider for instance the Betweenness Problem from Example~\ref{expl:betw-1}), so our result appears to be incomparable to Courcelle's. 

In the proof of our result, we need the following concepts from model theory.   A first-order theory $T$ is called \emph{$\omega$-categorical} if all countable models of $T$ are isomorphic~\cite{HodgesLong}. 
  A structure $\bB$ is called $\omega$-categorical
  if its first-order theory (i.e., the set of first-order sentences that hold in $\bB$) is $\omega$-categorical. Note that with this definition, finite structures are $\omega$-categorical. Another classic example is the structure $({\mathbb Q};<)$. The definition of treewidth can be treated as a black box in our proof, and we refer the reader to~\cite{BodDalJournal}.  

\begin{theorem}\label{thm:btw}
Let $\Phi$ be an ASNP or a connected GMSNP $\tau$-sentence and let $k \in \mathbb N$. Then the problem
to decide whether a given finite $\tau$-structure $\bA$ of treewidth at most $k$ satisfies $\Phi$ can be decided in polynomial time with a Datalog program (of width $k$). 
\end{theorem}
 \begin{proof}
  Since structures that are homogeneous in
 a finite relational language are $\omega$-categorical~\cite{HodgesLong} and first-order
 reducts of $\omega$-categorical structures are $\omega$-categorical~\cite{HodgesLong}, 
 Theorem~\ref{thm:first} and 
 Theorem~\ref{thm:gmsnp} 
 imply that 
 the problem to decide whether a finite
 $\tau$-structure satisfies $\phi$ can be formulated
 as CSP$(\bB)$ for an $\omega$-categorical structure $\bB$. 
 Then the statement 
 follows from Corollary~1 in~\cite{BodDalJournal}. 
 \qed \end{proof} 
 
\begin{remark} In Theorem~\ref{thm:btw} it actually suffices to assume that the \emph{core} of $\bA$ has treewidth at most $k$.
\end{remark}
 
 \begin{corollary}\label{cor:gsnp}
Let $\Phi$ be a GMSNP $\tau$-sentence and let $k \in {\mathbb N}$. Then 
 there is a polynomial-time algorithm that decides whether a given $\tau$-structure of treewidth at most $k$ satisfies $\Phi$.  
 \end{corollary}
 
\begin{proof}
 By Theorem~\ref{prop:disconnected}
there are finitely many connected
GMSNP sentences $\Phi_1$, \dots, $\Phi_k$ such that 
$\llbracket \Phi \rrbracket = \bigcup_{i \leq k} \llbracket \Phi_i \rrbracket$. Theorem~\ref{thm:btw} 
implies that there is a polynomial-time algorithm to decide whether a given $\bA$
of treewidth at most $k$ satisfies $\Phi_i$, for every $i \leq k$. Hence, $\bA$ satisfies $\Phi$ if and only if this algorithm accepts $\Phi_i$ for some $i \leq k$. \qed 
 \end{proof}

 \section{Conclusion and Open Problems}
 ASNP is a candidate for an
 expressive logic with a complexity dichotomy: 
 every problem in ASNP is NP-complete or in P if and only if the infinite-domain dichotomy conjecture for \red{first-order} reducts of binary finitely bounded homogeneous structures holds. See Figure~\ref{fig:SNP}
 for the relation to other candidate logics 
 that are known to have a dichotomy, 
 might have a complexity, or provably do not have 
 a dichotomy. 
 \begin{figure}
 \begin{center}
 \includegraphics[scale=.4]{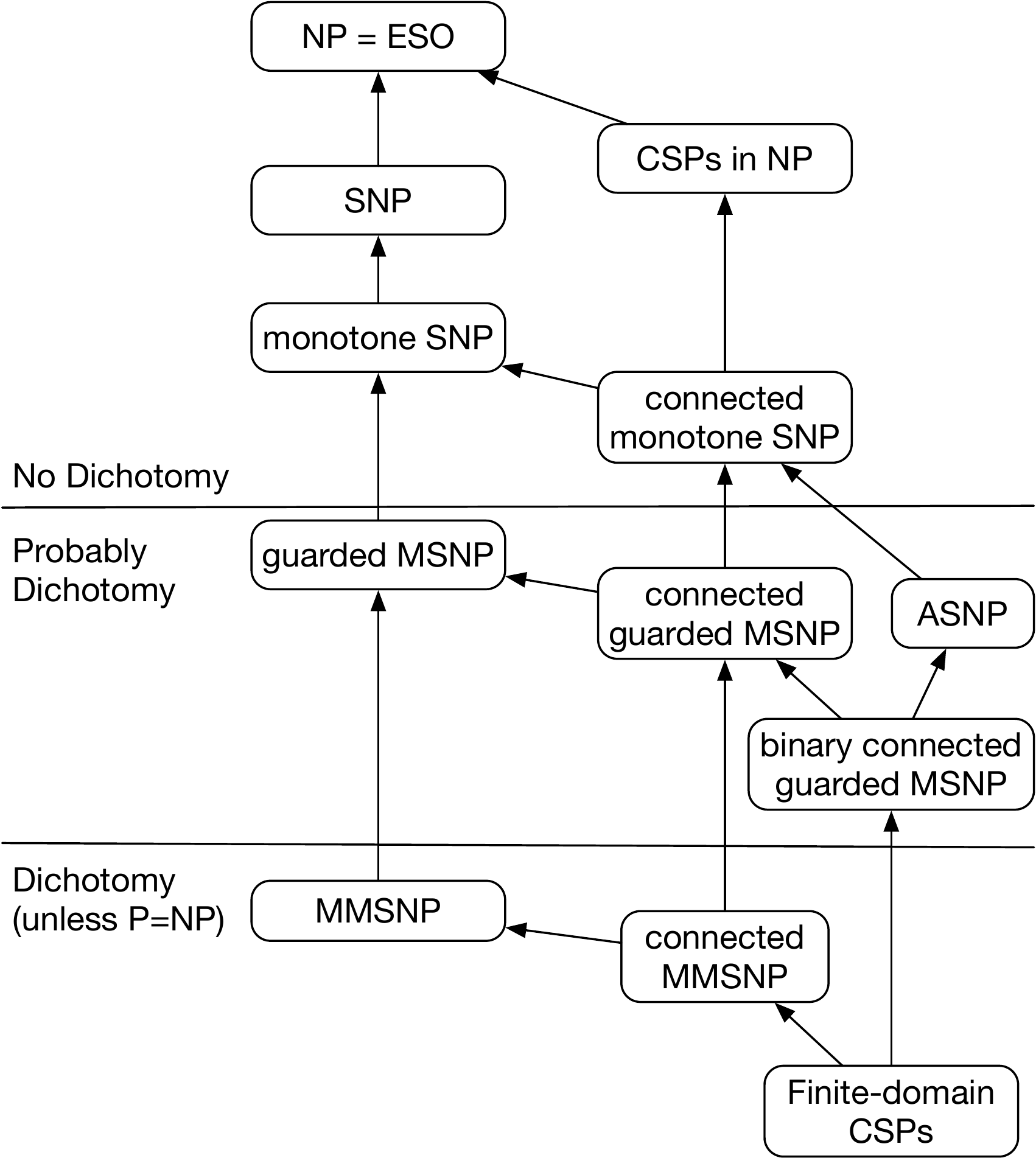}
 \end{center}
\caption{Fragments of existential second-order logic and complexity dichotomies.}
\label{fig:SNP}
 \end{figure}
 We presented an application of ASNP concerning
 the evaluation of computational problems on classes of structures of bounded treewidth. 
 We also proved 
 that the syntax of ASNP is algorithmically decidable. 
 The following problems concerning ASNP are open. 

 \begin{enumerate}
\item Is the Amalgamation Property decidable for 
(not necessarily binary) classes given by finitely many forbidden substructures? 
 \item 
 Is every binary CSP in Monadic Second-Order Logic (MSO) also in ASNP? 
 \item Can we decide algorithmically whether a given ASNP sentence is equivalent (over finite structures) to a fixed-point logic sentence (this then implies that the problem is in P)? We refer to~\cite{RydvalFP}
 for a recent article on the power of fixed-point logic for infinite-domain CSPs. 
 \item Is every problem in NP polynomial-time equivalent to a problem in Amalgamation SNP (without monotonicity)? 
 \item Is there a natural logic (which in particular has an effective syntax) that contains both ASNP and connected GMSNP and which describes CSPs for reducts of finitely bounded homogeneous structures? 
 \end{enumerate}
 
\bibliographystyle{abbrv} 
\bibliography{../../../global}

\newpage
\appendix

\section{Proofs for Connected Monotone SNP}
\label{sect:a1}
To prove Theorem~\ref{thm:connected}, suppose first that $\Phi$ is a connected monotone SNP sentence.
To show that $\Phi$ describes a problem of the form $\Csp(\bB)$ it suffices to show that the class of structures that satisfy $\Phi$ is closed under disjoint unions and inverse homomorphisms.
Let $\Phi$ be of the form $\exists R_1,\dots,R_k \; \forall x_1,\dots,x_l \colon  \phi$ where $\phi$ is a quantifier-free first-order $\sigma$-formula
where $\sigma := \tau \cup \{R_1,\dots,R_k\}$.

Suppose that $\bA_1$ and $\bA_2$ are $\tau$-structures that satisfy $\Phi$.
In other words, there is a $\sigma$-expansion $\bA_1^*$ of $\bA_1$ 
and a $\sigma$-expansion $\bA_2^*$ of $\bA_2$
such that these expansions satisfy $\forall \bar x \colon \phi$.
We claim that the disjoint union $\bA^*$ of $\bA_1^*$ and $\bA_2^*$ also satisfies $\forall \bar x \colon \phi$;
otherwise, there would be a clause $\psi$ in $\phi$ and elements $a_1,\dots,a_q$
of $A_1 \cup A_2$ such that $\psi(a_1,\dots,a_q)$ is false in $\bA^*$.
Since $\bA_1^*$ and $\bA_2^*$ satisfy $\forall \bar x \colon \psi$,
there must be $i,j$ such that $a_i \in A_1$ and $a_j \in A_2$. But then $\psi$ is disconnected,
a contradiction. 
Closure under inverse homomorphism follows
from monotonicity. 

For the second part of the statement, 
suppose that $\Phi$ describes a problem
of the form $\Csp(\bB)$ for some infinite structure $\bB$.
In particular, the class of structures that satisfy $\Phi$ is closed
under inverse homomorphisms. 
Then it follows from results of Feder and Vardi~\cite{FederVardiNegation} that $\Phi$ is equivalent to a monotone
SNP sentence. Moreover, the class of structures that satisfy $\Phi$ is closed under disjoint unions. 
Consider the SNP sentence $\Psi = \exists R_1,\dots,R_k,E \; \forall x_1,\dots,x_l\colon  \psi$ where $\psi$ is the conjunction 
of the following clauses (we assume without loss of generality that $l \geq 3$). 
\begin{itemize}
\item[a] For each relation symbol $R \in \tau$, say of arity $p$, 
and each $i<j\leq p$, add the conjunct $\neg R(x_1,\dots,x_p) \vee E(x_i,x_j)$ to $\psi$.
\item[b] Add the conjunct $\neg E(x_1,x_2) \vee \neg E(x_2,x_3) \vee E(x_1,x_3)$ 
to $\psi$.
\item[c] Add the conjunct $\neg E(x_1,x_2) \vee E(x_2,x_1)$ to $\psi$.
\item[d] For each clause 
$\phi'$ of $\phi$ with variables $y_1,\dots,y_q \subseteq \{x_1,\dots,x_l\}$, add to $\psi$ the conjunct $$\phi' \vee \bigvee_{i<j \leq q} \neg E(y_i,y_j)\; .$$
\end{itemize}
Clearly, $\Psi$ is monotone if $\Phi$ is monotone. 
We claim that the connected monotone SNP sentence $\Psi$ is equivalent to $\Phi$. 
Suppose first that $\bA$ is a finite structure that satisfies $\Phi$.
Then there is a $\sigma$-expansion $\bA'$ of $\bA$ that satisfies $\forall \bar x \colon \phi$.
The expansion of $\bA'$ by the relation $E = A^2$ shows that $\bA$ also satisfies $\forall \bar x \colon \psi$.

Now suppose that $\bA$ is a finite structure with domain $A$ that satisfies $\Psi$.
Then there is a $(\sigma \cup \{E\})$-expansion $\bA'$ of $\bA$ that satisfies $\forall \bar x \colon \psi$.
Write $\bA' = \bA'_1 \uplus \dots \uplus \bA_l'$ for connected $\sigma$-structures $\bA'_1, \dots, \bA_l'$.
Note that the clauses of $\psi$ force that the relation $E$ denotes $A_i^2$ in the structure $\bA'_i$, for each $i \leq l$.
Let $\bA_i$ be the $\sigma$-reduct of $\bA'_i$. Then $\bA_i$ satisfies $\forall \bar x \colon \phi$, because if there was a clause
$\phi'$ from $\phi$ violated in $\bA_i$ then the corresponding clause in $\psi$ would be violated in $\bA_i'$.
Hence, $\bA_i \models \Phi$ for every $i \leq l$, and since $\Phi$ is closed under disjoint unions,
we also have that $\bA \models \Phi$. 
\qed

\section{From CSPs to ASNP Sentences}
\label{sect:csp-ASNP}
Let $\bC$ be a \red{first-order} reduct of a binary finitely bounded homogeneous structure $\bB$. 
In this section we construct an ASNP sentence
such that $\Csp(\bC) = \llbracket \Phi \rrbracket$,
thus proving Theorem~\ref{thm:find-asnp}. 

Let $\sigma$ be the signature of
$\bB$ and $\tau$ the signature of $\bC$. 
We may assume without loss of generality
that $\bB$ contains a binary relation $E$ that denotes the equality relation; it is easy to
see that an expansion by the equality relation
preserves finite boundedness. 
Consider the structure $\bB^*$ with the domain
$B \times {\mathbb N}$ where 
$$R^{\bB^*} := \{((b_1,n_1),\dots,(b_k,n_k)) \mid 
n_1,\dots,n_k \in {\mathbb N}, (b_1,\dots,b_k) \in R^{\bB} \}.$$
To show that $\bB^*$ is homogeneous, 
let $h$ be an isomorphism between finite 
substructures of $\bB^*$. Let $T \subseteq B$ be the set of all first entries of elements of the first structure. Define $g \colon T \to B$ by 
picking for $b \in T$ 
an element of the form $(b,n) \in S$
and defining by $g(b) := h(b,n)_1$. This is well-defined: if $h$ is defined on $(b,n_1)$ and on $(b,n_2)$, then $((b,n_1),(b,n_2)) \in E^{\bB^*}$, and hence $h(b,n_1)_1 = h(b,n_2)_1$. The same consideration for $h^{-1}$ shows that $g$ is a bijection, and in fact an isomorphism between
finite substructures of $\bB$. 
By the homogeneity of $\bB$ there exists 
an extension $g^* \in \Aut(\bB)$ of $g$.
For each $b \in B$ pick a permutation 
$f_b$ of ${\mathbb N}$ that extends the bijection
given by $n \mapsto h(b,n)_2$. 
Then the map $h^* \colon B^* \to B^*$
given by $h(b,n) := (g^*(b),f_b(n))$ is an automorphism of $\bB^*$ that extends $h$. 
Since $\bB$ is finitely bounded, there exists
a universal $\sigma$-formula $\phi$ such that
$\Age(\bB) = \llbracket \phi \rrbracket$. 
Note that $\phi$ might contain the equality symbol (which we do not allow in SNP sentences).
Let $\phi^*$ be the formula obtained from $\phi$ by 
\begin{itemize}
\item replacing each occurrence of the equality symbol by the symbol $E \in \sigma$; 
\item joining conjuncts that imply that $E$ denotes an equivalence relation; 
\item joining for every $R \in \sigma$ of arity $n$ the conjunct 
$$\forall x_1,\dots,x_n,y_1,\dots,y_n \big (R(x_1,\dots,x_n) \vee \neg R(y_1,\dots,y_n) 
\vee \bigvee_{i \leq n} \neg E(x_i,y_i) \big)$$
(implementing indiscernibility of identicals for 
the relation $E$). 
\end{itemize}
We claim that $\Age(\bB^*) = \llbracket \phi^* \rrbracket$. To see this, let $\bA^*$ be a finite $\sigma$-structure. 
If $\bA^*$ satisfies $\phi^*$, then every induced
substructure $\bA$ of $\bA^*$ with the property
that $(x,y) \in E^{\bA}$ implies that at most one of $x$ and $y$ is an element of $A$, 
satisfies $\phi$, and hence is a substructure
of $\bB$. This in turn means that $\bA^*$ is in $\Age(\bB^*)$. The implications in this statement can be reversed which shows the claim. 

\medskip 
Let $\phi'$ be the formula obtained from $\phi^*$ 
as follows. \red{For each $S \in \tau$ let
$\chi_S$ be the first-order definition of 
$S^{\bC}$ in $\bB$; since $\bB$ is homogeneous
we may assume that $\chi_S$ is quantifier-free~\cite{HodgesLong}.
Furthermore, we may assume that $\chi_S$ is given in conjunctive normal form. Let
$k$ be the arity of $S$. 
We then add for each conjunct $\chi_{S}'$
of $\chi_S$ the conjunct 
$$\forall x_1,\dots,x_k \big (S(x_1,\dots,x_k) \Rightarrow \chi'_S(x_1,\dots,x_k) \big)$$
By construction, the sentence $\Phi$ obtained from $\phi'$ 
by quantifying all relation symbols of $\sigma$ 
is an ASNP $\tau$-sentence. }

 \red{We claim that a finite $\tau$-structure $\bA$ 
 satisfies $\Phi$ if and only if $\bA \rightarrow \bC$.
 If $\bA \models \Phi$ then $\bA$ has a
 $\tau \cup \sigma$-expansion $\bA'$ 
 that satisfies $\phi'$. Let $\bA^*$ be the $\sigma$-reduct of $\bA'$. By the construction of $\Phi$ this means that $\bA^*$ is in $\Age(\bB^*)$. 
Since $\bA'$ satisfies the additional conjuncts 
of $\phi'$ we have 
 $\bA \rightarrow \bC$.
Conversely, suppose that $\bA \rightarrow \bC$. Then $\bA \rightarrow \bB^*$. Let $\bA'$ be the 
$\tau \cup \sigma$-structure defined as in Claim 1 of the previous theorem. Then $\bA' \models \phi'$ and therefore witnesses $\bA \models \Phi$.}
%
\qed 

\section{Proofs Concerning Guarded SNP} 
\label{sect:a4}
We prove Proposition~\ref{prop:disconnected}. 
Let $\Phi$ be a guarded SNP sentence. 
Suppose that 
the quantifier-free part of $\Phi$
has a disconnected clause $\psi$ (Definition~\ref{def:mmc}).
By definition the variable set can be partitioned
into non-empty variable sets $X_1$ and $X_2$
such that for every negative literal
$\neg R(x_1,\dots,x_r)$ of the clause 
either $\{x_1,\dots,x_r\} \subseteq X_1$ or $\{x_1,\dots,x_r\} \subseteq X_2$. 
The same is true for every positive literal,
since otherwise the definition of guarded clauses would imply a negative literal on a set that
contains $\{x_1,\dots,x_r\}$,
contradicting the property above. 
Hence,  
 $\psi$ can 
be written as $\psi_1(\bar x) \vee \psi_2(\bar y)$ 
for non-empty disjoint tuples of variables $\bar x$ and $\bar y$. 
Let $\phi_1$ be the formula obtained from 
$\phi$ by replacing $\psi$ by $\psi_1$,
and let $\phi_2$ be the formula obtained from 
$\phi$ by replacing $\psi$ by $\psi_2$.

Let $P_1,\dots,P_k$ be the existential predicates in $\Phi$, and let $\tau$ be the input signature of $\Phi$. 
It suffices to show that for every $(\tau \cup \{P_1,\dots,P_k\})$-expansion 
$\bA'$ of $\bA$ we have that $\bA'$ satisfies $\phi$ if and only if $\bA'$ satisfies $\phi_1$ or $\phi_2$. 
If $\bA'$ falsifies a clause of $\phi$,
there is nothing to show since then
$\bA'$ satisfies neither $\phi_1$ nor $\phi_2$.
If $\bA'$ satisfies all clauses of $\phi$,
it in particular satisfies a literal from $\psi$;
depending on whether this literal lies in $\psi_1$ or in $\psi_2$, we obtain that $\bA'$ satisfies
$\psi_1$ or $\psi_2$, and hence
$\phi_1$ or $\phi_2$.
Iterating this process for each 
disconnected clause of $\phi$, 
we eventually arrive at a finite 
disjunction of connected
guarded SNP sentences. 
\qed

 \end{document}